\documentclass[12pt,a4paper]{article}
\usepackage[utf8]{inputenc}
\usepackage[english]{babel}
\usepackage{amsmath}
\usepackage{dsfont}
\usepackage{amsfonts}
\usepackage{amssymb}
\usepackage{multicol}
\usepackage{graphicx}
\usepackage{titlesec}
\usepackage{url} 
\usepackage{authblk} 
\usepackage[bookmarksnumbered, colorlinks, plainpages]{hyperref}
\usepackage[left=2cm,right=2cm,top=2cm,bottom=2cm]{geometry}
\numberwithin{equation}{section}

\newtheorem{definition}{Definition}[section]
\newtheorem{theorem}{Theorem}[section]

\newtheorem{assump}{Assumption}[section]
\newtheorem{proof}{Proof}[section]

\usepackage[left=2cm,right=2cm,top=2cm,bottom=2cm]{geometry}

\everymath{\displaystyle}

\titleformat{\section}
{\centering\large\bfseries}{\thesection.}{0.5em}{}

\titleformat{\subsection}
{\centering\normalsize\bfseries}{\thesubsection.}{0.5em}{}

\newcommand{\Ni}{\noindent}

\title{\textbf{Wavelet-based estimation of long-memory parameter in stochastic volatility models using a robust log-periodogram}}

\author[1]{Manganaw N'Daam}
\author[1]{Tchilabalo Abozou Kpanzou}
\author[1]{Edoh Katchekpele}

\affil[1]{Laboratoire de Modélisation Mathématique et d’Analyse Statistique Décisionnelle (LaMMASD)\\
	Département de Mathématiques, Université de Kara, BP 404 Kara, TOGO.}

\affil[ ]{\textit{Emails}: \texttt{manganawn@gmail.com}, \texttt{kpanzout@gmail.com}, \texttt{edohkatchekpele@gmail.com}}

\date{}
\begin{document}
	
	\maketitle
	\noindent
	\author{}

	\Ni \textbf{Abstract:} 
	In this paper, we propose a novel method for estimating the long-memory parameter in time series. By combining the multi-resolution framework of wavelets with the robustness of the Least Absolute Deviations (LAD) criterion, we introduce a periodogram providing a robust alternative to classical methods in the presence of non-Gaussian noise. Incorporating this periodogram into a log-periodogram regression, we develop a new estimator. Simulation studies demonstrate that our estimator outperforms the  Geweke and Porter-Hudak (GPH) and Wavelet-Based Log-Periodogram (WBLP) estimators, particularly in terms of mean squared error, across various sample sizes and parameter configurations.\\
	
	\Ni \textbf{Key Words and Phrases}: Long-memory parameter estimation, Wavelet-based methods, Log-periodogram regression, Simulation.
	
	\section{Introduction}
	The estimation of the long-memory parameter is a fundamental problem in time
	series analysis, particularly for long-memory processes where the autocorrelation function decays hyperbolically. These processes are prevalent in various
	fields, including econometrics, finance, climatology, and social sciences, where
	accurate characterization of long-range dependence is crucial for effective modeling and forecasting.\\
	
	\Ni Since the seminal work of \cite{hosking}, numerous methods have been proposed for estimating the long-memory parameter, denoted \( d \). One of the most widely used approaches is the semi-parametric log-periodogram regression method introduced by \cite{gph}. Despite its popularity, this method suffers from significant limitations, particularly its sensitivity to non-Gaussian noise, which can induce substantial negative bias in the estimation, as noted by \cite{lee1}. Addressing these shortcomings has been a key focus of recent research efforts.\\
	
	\Ni Wavelet-based methods have emerged as a promising alternative, offering superior robustness by leveraging both time and frequency localization properties.
	\cite{lee1} introduced an estimator based on the wavelet periodogram, demonstrating improved performance under non-Gaussian conditions. More recently, \cite{xu2022} proposed a wavelet-based approach for estimating \( d \) in potentially non-linear and non-Gaussian time series, further highlighting the advantages of wavelet techniques. Similarly, \cite{pinto2023} developed an estimator based on fractional spline wavelets, reinforcing the effectiveness of wavelet-based estimation in capturing long-memory behavior.\\
	
	\Ni In this paper, we propose a novel estimation method for the long-memory parameter in long-memory stochastic volatility (LMSV) models. Our approach builds upon the wavelet log-periodogram framework, introducing a periodogram, which relies on the least absolute deviations (LAD) method applied to wavelet coefficients. This modification enhances robustness against noise and improves estimation accuracy, particularly in financial time series characterized by heavy-tailed distributions and volatility clustering.\\
	
	\Ni The remainder of this paper is structured as follows. Section \ref{sec_02} provides an overview of wavelet analysis and its applications in time series modeling. Section \ref{sec_03} reviews key estimators of the long-memory parameter, with a focus on the contributions of  \cite{gph} and  \cite{lee1}. Section \ref{sec_04} presents the new periodogram and details our proposed estimator. Section \ref{sec_05} evaluates the performance of our method through a comparative simulation study, demonstrating its advantages over existing techniques. Finally, Section \ref{sec_06} gives some concluding remarks.
	
	\section{Wavelet Analysis} \label{sec_02}
	In this section, we introduce the discrete wavelet transform. This method relies on the construction of orthonormal bases obtained by dyadic dilation and translation of a pair of specific functions, $\phi$ and $\psi$, called the father wavelet and mother wavelet, respectively, such that:
	\[
	\int_{-\infty}^\infty \phi(x) \, dx = 1, \quad \int_{-\infty}^\infty \psi(x) \, dx = 0.
	\]
	
	\noindent In time series analysis, the father wavelet \(\phi(\cdot)\) is used to represent smooth and low-frequency components of the series, while the mother wavelet \(\psi(\cdot)\) captures details and high-frequency components.\\
	
	\noindent The wavelet family \( \{\psi_{j,q}, \, (j, q) \in \mathbb{Z} \times \mathbb{Z}\} \) and scaling functions \( \{\phi_{j,q}, \, (j, q) \in \mathbb{Z} \times \mathbb{Z}\} \) are constructed by performing dilation and translation operations on the functions \( \phi \) and \( \psi \), defined as:
	\[
	\psi_{j,q}(t) = 2^{\frac{j}{2}} \psi(2^j t - q), \quad \phi_{j,q}(t) = 2^{\frac{j}{2}} \phi(2^j t - q),
	\]
	where \(j = 1, \dots, J\) represents the scale, and \(q = 1, \dots, 2^j\) corresponds to the translation. 
	The parameter \(j\) dilates the wavelet functions to adjust their support and capture characteristics of low or high frequencies, while \(q\) performs temporal translation. The maximum number of scales \(J\) depends on the number of observations \(n\), with the constraint \(n \geq 2^J\). These families are used to compute wavelet and scaling coefficients, respectively.\\
	
	\noindent An important property of wavelets is localization: the coefficient associated with \(\psi_{j,q}(t)\) provides local information about the function around the approximate position \(q 2^{-j}\) and frequency \(2^j\).\\
	
	\noindent Any function belonging to \(L^2(\mathbb{R})\) can be developed over a wavelet basis as follows:
	\[
	X(t) = \sum_q a_{J_0, q} \phi_{J_0, q}(t) + \sum_{j > J_0} \sum_q b_{j,q} \psi_{j,q}(t),
	\]
	where \(\phi_{J_0,q}\) is a scaling function associated with coarse coefficients \(a_{J_0,q}\), and \(\psi_{j,q}\) is associated with fine coefficients \(b_{j,q}\), all given by the relations:
	\[
	a_{J_0, q} = \int X(t) \phi_{J_0, q}(t) \, dt, \quad b_{j,q} = \int X(t) \psi_{j,q}(t) \, dt.
	\]
	\noindent See \cite{beals2004} for \(L^2(\mathbb{R})\) functions.\\
	
	\noindent These coefficients measure the respective contributions of the scaling functions and wavelets to reconstruct \(X(t)\). This decomposition allows \(X(t)\) to be represented as orthogonal components at different resolutions, constituting a multi-resolution analysis (MRA). For more details, see \cite{mallat}.\\
	
	\noindent We now explicitly present the properties of wavelet functions (see \cite{lee1}).\\
	\begin{assump}\label{assump2.1}
		\begin{itemize}
			\item[(a)] $\psi : \mathbb{R} \to \mathbb{R}$ such that:
			\[
			\int_{-\infty}^\infty \psi(x) \, dx = 0, \quad 
			\int_{-\infty}^\infty |\psi(x)| \, dx < \infty, \quad 
			\int_{-\infty}^\infty (1 + x^2) |\psi(x)| \, dx < \infty.
			\]
			
			\item[(b)] $|\widehat{\psi}(\lambda)| = \vert \lambda \vert^\nu b(\lambda)$, with $\dfrac{b(t\lambda)}{b(\lambda)} = 1$ for all $t$, as $\lambda \to 0$, where $\nu$ is a positive integer, $0 < b(0) < \infty$, and $\hat{\psi}(\lambda)$ is the Fourier transform of $\psi$, defined as:
			\[
			\hat{\psi}(\lambda) = (2\pi)^{-1/2} \int_{-\infty}^\infty \psi(x) e^{-i \lambda x} \, dx.
			\]
		\end{itemize}
	\end{assump}
	
	\noindent Assumption \ref{assump2.1}(a) describes the properties of the wavelet function $\psi\left(\cdot\right)$. Assumption  \ref{assump2.1}(b), on the other hand, models the spectral behavior of the Fourier transform around $\lambda = 0$. The integer parameter $\nu$ corresponds to the number of vanishing moments of $\psi\left(\cdot\right)$, meaning:
	\[
	\int_{-\infty}^\infty x^r \psi(x) \, dx = 0, \quad \text{for } r = 0, 1, \dots, \nu-1.
	\]
	\noindent The vanishing of the first $\nu$ moments is equivalent to the vanishing of the first $\nu$ spectral derivatives at zero, that is:
	\[
	\frac{d^r}{d\lambda^r} \hat{\psi}(\lambda) \Big|_{\lambda=0} = 0, \quad \text{for } r = 0, 1, \dots, \nu-1.
	\]\\
	
	\noindent This hypothesis is satisfied if $\psi(\cdot)$ has a compact support and belongs to the class \(C^\nu(\mathbb{R})\), meaning that all derivatives up to order $\nu$ exist and the $\nu$-th derivative \(f^{(\nu)}\) is continuous on \(\mathbb{R}\).\\
	
	\noindent In practical applications, time series are generally discrete rather than continuous. Consequently, instead of using continuous wavelets, discrete sequences called wavelet filters are employed (see \cite{daubechi}). The length of these sequences corresponds to the width of the wavelet filter. This filtering-based approach makes wavelet analysis particularly suitable for time series.\\
	
	\noindent In the discrete wavelet transform (DWT), wavelet coefficients are obtained through the multi-resolution analysis (MRA) scheme, which is implemented using a two-channel filter bank representing the wavelet transform, see \cite{mallat}.\\
	
	\noindent \cite{daubechi} introduced an important class of wavelet filters called compactly supported Daubechies wavelet filters, which have the smallest support for a given number of vanishing moments. Two types of filters are distinguished: the extremal phase filters \(D(L)\) and the least asymmetric filters \(LA(L)\).\\
	
	\noindent An improved version of the DWT is the maximal overlap discrete wavelet transform (MODWT). As explained by \cite{percival}, the MODWT algorithm follows the same filtering steps as the standard DWT but without down-sampling, i.e., without halving the resolution. Thus, the number of scaling and wavelet coefficients at each level of the transform remains equal to the total number of observations in the sample. For more details, see \cite{percival}.\\
	
	\noindent After examining the fundamental principles of wavelet analysis, we now turn our attention to estimators of the long-memory parameter.

	\section{Long-Memory Parameter Estimators} \label{sec_03}
	In this section, we present two estimators for evaluating the long-memory parameter. Each method relies on specific assumptions and tools, offering advantages and limitations depending on the application context. The goal is to compare these approaches and analyze their relevance for different types of processes, considering their robustness and accuracy.
	
	\subsection{GPH Estimator}\label{subsec_1}
	In this subsection, we present the estimation method developed by \cite{gph}, which is based on a least squares regression in the spectral domain, leveraging the spectral density behavior \(f(\cdot)\) of a long-memory process at the origin:\\
	
	\begin{assump}\label{assump3.1}
		\begin{equation}
			f(\lambda) = \vert \lambda \vert^{-2d}f_{u}(\lambda), \quad \text{as } \lambda \rightarrow 0, \label{densite}
		\end{equation}
		where \(d\) is the long-memory parameter and \(f_{u}(\lambda)\) is a function assumed to be finite, bounded away from zero and continuous on the interval \(\left[-\pi,\pi\right]\).
	\end{assump}
	
	\noindent Consider the problem of estimating the parameter \(d\) in the general integrated long-memory process $\left\{X_{t}\right\}$ (see \cite{hosking}). Assume that:
	\begin{equation*}
		(1 - B)^d X_t = u_t, 
	\end{equation*}
	where $B$ is the backward shift operator, and $\left\{u_t\right\}$ is a stationary linear process with a spectral density function \(f_u(\lambda)\) that is finite, bounded away from zero, and continuous on the interval \(\left[-\pi, \pi\right]\).\\
	
	\noindent The spectral density $f(\cdot)$ of the integrated long-memory process $\left\{X_{t}\right\}$ is given by:
	\begin{equation}
		f(\lambda) = \left[4 \sin^{2}\left(\frac{\lambda}{2}\right)\right]^{-d}f_{u}(\lambda) .\label{densite_arfima}
	\end{equation}
	
	\noindent Exploiting these spectral properties, \cite{gph} propose a log-periodogram regression method for semi-parametrically estimating $d$ based on the first harmonic frequencies. While widely adopted for its efficiency and simplicity, this approach can be sensitive to certain violations of assumptions, such as the presence of outliers, heavy-tailed distributed errors (see \cite{lim2018} and \cite{lim2022}), or non-Gaussian noise.\\
	
	\noindent Applying the logarithm to (\ref{densite_arfima}), we obtain:
	\begin{equation}
		\log f\left(\lambda\right) = \log  f_{u}\left(0\right) -d\log  \left[4\sin^{2}\left(\frac{\lambda}{2}\right)\right] + \log\left[\frac{f_{u}\left(\lambda\right)}{f_{u}\left(0\right)}\right]. \label{3}
	\end{equation}
	
	\noindent Let $I_{n}\left(\lambda_{k}\right)$ denote the ordinary periodogram of the integrated long-memory process $\left\{X_{t}\right\}$ evaluated at the Fourier frequencies $\lambda_{k} = \frac{2\pi k }{n}$, with $ k = 1,2,\cdots,m$, where $n$ is the sample size and $m$ is the number of Fourier frequencies considered. Adding $\log I_{n}\left(\lambda_{k}\right)$ to both sides of (\ref{3}), we obtain:
	\begin{equation}
		\log  I_{n}\left(\lambda_{k}\right) = \log  f_{u}\left(0\right) -d\log  \left[4\sin^{2}\left(\frac{\lambda_{k}}{2}\right)\right]  + \log \left[\frac{I_{n}\left(\lambda_{k}\right)}{f\left(\lambda_{k}\right)}\right] + \log \left[\frac{f_{u}\left(\lambda_{k}\right)}{f_{u}\left(0\right)}\right], \label{4}
	\end{equation}
	\noindent where the ordinary periodogram at frequency $\lambda_{k}$ is defined as: 
	$$ I_{n}\left(\lambda_{k}\right) = \frac{1}{2\pi n} \left| \sum_{t=1}^{n}X_{t}e^{it\lambda_{k}}\right|^{2} \quad \text{with} \quad k = 1,2,\cdots,m.  $$\\
	
	\noindent The log-periodogram estimator relies on the following hypothesis:\\
	
	\begin{assump}\label{assump3.2}
		\begin{itemize}
			\item [(a)] For low frequencies, the term $\log \left[\frac{f_{u}\left(\lambda_{k}\right)}{f_{u}\left(0\right)}\right]$ is negligible.
			\item [(b)] The random variables $\log \left[\frac{I_{n}\left(\lambda_{k}\right)}{f\left(\lambda_{k}\right)}\right]$, $k = 1,2,\cdots,m$, are asymptotically $i.i.d.$
		\end{itemize}
	\end{assump}
	
	\noindent Under Assumption \ref{assump3.2}, \cite{gph} rewrite equation (\ref{4}) as:
	\begin{equation*}
		Z_{k} = c_{0}   + dR_{k} + \epsilon_{k}, 
	\end{equation*}
	where $Z_{k} = \log  I_{n}\left(\lambda_{k}\right)$, $c_{0} =  \log  f_{u}\left(0\right) + C$, $R_{k} = -\log  \left[4\sin^{2}\left(\frac{\lambda_{k}}{2}\right)\right]$, $\epsilon_{k} = \log \left[\frac{I_{n}\left(\lambda_{k}\right)}{f\left(\lambda_{k}\right)}\right] - C$, and $\epsilon_{k}\sim i.i.d(-C,\frac{\pi^{2}}{6})$ with $C=-0.5772$ being the Euler constant. The ordinary least squares estimator is then given by:
	\begin{equation*}
		\hat{d}_{GPH} = \frac{\sum_{k=1}^{m}\left(R_{k}-\overline{R}_{m}\right)Z_{k}}{\sum_{k=1}^{m}\left(R_{k}-\overline{R}_{m}\right)^2},
	\end{equation*}
	where $\overline{R}_{m} = \frac{1}{m}\sum_{k=1}^{m}R_{k}.$\\
	
	\begin{assump}\label{assump3.3}
		$m = m(n) \rightarrow \infty$, $n\rightarrow \infty$, and $ \frac{m(n)}{n}\xrightarrow[n\rightarrow\infty]{}  0 $.
	\end{assump}
	
	\noindent Under Assumption \ref{assump3.3}, \cite{gph} prove the asymptotic normality of the estimator $\hat{d}_{GPH}$ when $\vert d \vert < \frac{1}{2}$:
	\begin{equation*}
		m^{\frac{1}{2}}(\hat{d}_{GPH}-d) \xrightarrow{D} \mathcal{N}\left(0,\frac{\pi^{2}}{6}\left[\sum_{k=1}^{m}\left(R_{k}-\overline{R}_{m}\right)^2\right]^{-1} \right) \text{ as } n \rightarrow \infty,
	\end{equation*}
	\noindent where $\xrightarrow{D}$ denotes the convergence in distribution (see \cite{lo2021} for details).\\
	
	\noindent Although the classical log-periodogram estimator is widely used to estimate the long-memory parameter, it can be significantly affected by the presence of noise or complex data structures. To overcome these limitations, wavelet-based approaches provide a powerful alternative by leveraging the time-frequency localization properties of wavelets.
	\subsection{Wavelet-Based Log-Periodogram Estimator}
	
	\noindent \cite{lee1} proposed a significant improvement to the semi-parametric log-periodogram regression method for estimating the memory parameter in long-memory stochastic volatility (LMSV) models. Although the log-periodogram estimator of \cite{gph} is widely used, it has important limitations. In particular, it violates the Gaussianity or martingale assumptions, resulting in a pronounced negative bias due to the presence of a non-Gaussian noise spectrum. \cite{lee1} demonstrated that applying the wavelet transform to the squared process effectively reduces the noise spectrum around zero frequency, yielding a quasi-Gaussian spectral representation at this frequency. Based on this, he proposed a regression estimator in the wavelet domain and established its asymptotic mean squared error and consistency in accordance with the asymptotic theory of long-memory processes developed by \cite{hdb}.\\
	
	\noindent Consider a long-memory stochastic volatility model for discrete time series \(\{X_{1},\cdots,X_{n}\}\), defined by:
	
	\begin{equation}
		X_t = \sigma \exp\left(\frac{Z_t}{2}\right) e_t, \quad t = 1,\cdots,n, \label{LMSV}
	\end{equation}
	
	\noindent where \(\{Z_t\}\) is a latent Gaussian long-memory process with a memory parameter \(d \in (0, 0.5)\), independent of the process \(\{e_t\}\), which is an i.i.d. noise process with zero mean. The spectral density of the process \(\{Z_t\}\) at zero frequency follows the form \eqref{densite}.\\
	
	\noindent The log-squared process is used as a measure of volatility and is expressed as:
	
	\[
	Y_t = \log(X_t^2) = \eta + Z_t + U_t,
	\]
	
	\noindent where \(\eta = \log \sigma^2 + \mathbb{E}(\log e_t^2)\) and \(U_t = \log e_t^2 - \mathbb{E}(\log e_t^2)\). Here, \(\{U_t\}\) is an i.i.d. process with zero mean and variance \(\sigma_U^2\). The autocovariances \(R(h)\) of \(\{Y_t\}\) are identical to those of \(\{Z_t\}\) for \(h \neq 0\).\\
	
	\noindent The spectral density of the process \(\{Y_t\}\) is thus the sum of the spectral density of the Gaussian long-memory process \(\{Z_t\}\) and that of the non-Gaussian noise \(\{U_t\}\), giving the following relationship:
	\begin{align*}
		f_Y(\lambda) &= f_{Z}(\lambda) + f_U(\lambda) \nonumber \\ 
		&= \vert \lambda \vert^{-2d} f_{u}(\lambda) + \frac{\sigma_U^2}{2\pi} \nonumber \\ 
		f_Y (\lambda)&=  \vert \lambda \vert^{-2d} \left( f_{u}(\lambda) + \frac{\sigma_U^2}{2\pi}  \vert \lambda \vert^{2d} \right), \quad \text{as} \quad \lambda \to 0, 
	\end{align*}
	
	\noindent where $\sigma_U^{2}$ is the variance of the non-Gaussian noise.\\
	
	\noindent Thus, this approach accounts for the long-memory effects while integrating the impact of non-Gaussian noise, as emphasized by \cite{lee1}.\\
	
	\noindent Due to the presence of the non-Gaussian noise spectrum, \cite{lee1} shows that the log-periodogram estimator suffers from a significant negative bias that behaves at the order of \(\vert\lambda\vert^{2d}\). To address this, he proposes using the wavelet transform of the squared process \(\{Y_t\}\) to obtain an approximately Gaussian spectral representation by effectively eliminating the noise spectrum around zero frequency.\\
	
	\noindent The discrete wavelet transform of the process $\left\{Y_{1},...,Y_{n}\right\}$ is defined as follows:\\
	\begin{equation}
		w_{jq} = 2^{\frac{j}{2}}\sum_{t=1}^{n}Y_{t}\psi\left(2^{j}t-q\right), \label{coeff_ondelette}
	\end{equation}
	
	\noindent where \( t \) is appropriately re-indexed to ensure full coverage of the wavelet support. For example, if the support of \( \psi \) is \([0, 1]\), then \( t = i/n \), for \( i = 1, 2, \dots, n \). The integers \( j \) and \( q \) are the scale and translation parameters, respectively, where \( j = 0, 1, \dots, J \) and \( q = 0, 1, \dots, 2^j - 1 \). The finest scale is set at \( J \), where \( n = 2^J \).\\
	
	\noindent Let Haar wavelet satisfies Assumption 1 with $\nu = 1 $ and be defined as:
	\begin{equation*}
		\psi(x) =
		\begin{cases} 
			1 & \text{if } 0 \leq x \leq 0.5, \\ 
			-1 & \text{if } 0.5 < x \leq 1, 
		\end{cases}
	\end{equation*} and its Fourier transform is:
	\begin{equation*}
		|\widehat{\psi}(\lambda)| = \dfrac{ \vert \lambda \vert}{4} \left[\frac{\sin^{2}(\lambda / 4)}{(\lambda / 4)^{2}}\right].
	\end{equation*}
	\noindent By applying the discrete wavelet transform to the process $\left\{Y_{1},...,Y_{n}\right\}$ with the Haar wavelet, it can be expressed as:
	\begin{equation*}
		w_{jq} = \alpha_{jq} + \beta_{jq}, 
	\end{equation*} 
	$\text{where} \quad \alpha_{jq} = 2^{\frac{j}{2}} \sum_{t=1/n} Z_t \psi(2^j t - q) \quad \text{and} \quad \beta_{jq} = 2^{\frac{j}{2}} \sum_{t=1/n} U_t \psi(2^j t - q).$\\
	
	\noindent \cite{lee1} shows that the spectral density of \( \beta_{jq} \) becomes zero at frequency zero. Let \( R_{\beta}(m) = \mathbb{E}[\beta_{jq} \beta_{jq+m}] \), the autocovariance of the transformed series \( \beta_{jq} \) at scale \( j \), and \( f_{\beta}^{(j)} \) the spectral density at scale \( j \). The wavelet transform \( \beta_{jq} \) is a linear combination of the i.i.d. noise process \( U_t \).\\
	
	\noindent When applying the Haar wavelet, it is simply the difference between local sums of \( U_t \) over the intervals \( t \in [2^{-j} q, 2^{-j} (q + 0.5)] \) and \( t \in (2^{-j} (q + 0.5), 2^{-j} (q + 1)] \). Furthermore, the sequence \( \beta_{jq} \) forms a $1$-dependent process, effectively behaving like a moving average (MA) process of order \( p \), with \( p \) determined by \( j \). As \( j \) increases, the width of the interval decreases. Ultimately, when \( j \) reaches the largest scale \( J \), the process \( \beta_{jq} \) approximates an \( MA(1) \) process, expressed as:
	\begin{equation*}
		\beta_{jq} = 2^{J/2} \left( U_{2^{-J} q} - U_{2^{-J} (q+1)} \right).
	\end{equation*}

	\noindent Define the autocovariance of \( \beta_{jq} \), denoted as \( R_{\beta}(m) = \mathbb{E}[\beta_{jq} \beta_{jq+m}] \). When \( j = J \), we obtain:  
	\begin{equation}
		f_{\beta}^{(J)}(0) = \frac{1}{2\pi} \sum_{m=-\infty}^{\infty} R_{\beta}(m) = 0, \label{spectral_density_beta}
	\end{equation}
	where  
	\begin{equation*}
		R_{\beta}(m) =
		\begin{cases} 
			2\sigma^2_U, & \text{if } m = 0, \\  
			-\sigma^2_U, & \text{if } m = \pm1, \\  
			0, & \text{if } |m| > 1.  
		\end{cases}  
	\end{equation*}
	
	\noindent Now, consider the spectral density of \( \alpha_{jq} \) at scale \( j \), denoted as \( f_{\alpha}^{(j)}(\lambda) \). Under Assumption 2, it is obtained directly as follows. The autocovariances of the wavelet coefficients at scale \( j \) are written as:  
	\begin{align*}
		\mathbb{E}[\alpha_{jq} \alpha_{j\tau}] &= 2^j \sum_t \sum_s \mathbb{E}[Z_t Z_s] \psi(2^j t - q) \psi(2^j s - \tau) \\ 
		&= 2^j \sum_t \sum_s \left( \int_{-\pi}^{\pi} f_Z(\lambda) e^{i(t-s)\lambda} d\lambda \right) \psi(2^j t - q) \psi(2^j s - \tau).
	\end{align*}
	
	\noindent This leads to:  
	\[
	\mathbb{E}[\alpha_{jq} \alpha_{j\tau}] = 2^{-j} \int_{-\pi}^{\pi} f_Z(\lambda) |\hat{\psi}(2^{-j} \lambda)|^2 e^{i2^{-j}(q-\tau)\lambda} d\lambda.
	\]  
	
	\noindent Thus, for \( j = J \), the spectral density of \( \alpha_{jq} \) is given by:  
	\begin{equation}
		f_{\alpha}^{(J)}(\lambda) = 2^{-J} f_Z(\lambda) |\hat{\psi}(2^{-J} \lambda)|^2, \quad \lambda \in [-\pi, \pi]. \label{spectral_density_alpha}
	\end{equation}
	
	\noindent Combining Equations \eqref{spectral_density_beta} and \eqref{spectral_density_alpha}, \cite{lee1} derives the approximate Gaussian spectral representation of the discrete wavelet transform around frequency zero: 
	\begin{equation}
		f_w(\lambda) = C_J \vert \lambda \vert^{-2(d-\nu)} f_{u}(\lambda) h(\lambda) \quad \text{as} \quad \lambda \to 0, \quad \text{for} \quad d \in (0, 0.5), \label{densite_avec_ond}
	\end{equation} 
	where $C_J = 2^{-J(1+2\nu)}$ term and \( h(\lambda) = b^2(\lambda) \). \\
	
	\noindent This spectral representation (\ref{densite_avec_ond}) provides a foundation for the semi-parametric estimation of \( d \). The spectral density \( f_w(\lambda) \) behaves as \( \vert \lambda \vert ^{-2(d-v)} \) around frequency zero, thus \( f_w(\lambda) \sim \vert\lambda\vert^{-2(d-1)} \) when the Haar wavelet is used. The functions \( f_u(\lambda) \) and \( h(\lambda) \) originate from the short-term dependence in $\left\{Z_{t}\right\}$ and the discrete wavelet transform, respectively. Note that these two functions are even, continuous over \( [-\pi, \pi] \), and non-zero at \(\lambda = 0\).\\
	
	\noindent In his work, \cite{lee1} extends the wavelet-based log-periodogram estimation method, drawing inspiration from \cite{hdb} to analyze its asymptotic properties. The approach constructs a periodogram for the wavelet transform at scale \( j \), defined as:  
	\begin{equation*}
		I_k = I_k^{(j)} = \frac{1}{2\pi n} \sum_{q=0}^{2^J-1} \left| w_{jq} \exp(i\lambda_k q) \right|^2, \quad k = 1, 2, \dots, m,
	\end{equation*}
	where \( \lambda_k = \frac{2\pi k}{n} \) represents the associated frequencies. \\
	
	
	\noindent Under the spectral representation \( f_w(\lambda) \) given in \eqref{densite_avec_ond} and by letting \( s(\lambda) = f_u(\lambda)h(\lambda) \), \cite{lee1} reformulates the log-periodogram regression as follows:  
	\begin{equation}
		\log I_k = \alpha + \gamma X_k + \log\left[\frac{s(\lambda_k)}{s(0)}\right] + \varepsilon_k, \quad k = 1, 2, \dots, m, \label{wavelet_LP}
	\end{equation}
	where:  
	\[
	\alpha = \log C_J + \log s(0), \quad \gamma = d - 1, \quad X_k = -2\log(\lambda_k), \quad \varepsilon_k = \log\left[\frac{I_k}{f_w(\lambda_k)}\right].
	\] 
	
	\noindent The wavelet-based log-periodogram estimator, denoted as \( \hat{d}_{WBLP} \), is obtained by applying a logarithmic transformation to equation (\ref{wavelet_LP}). Specifically, it involves performing a regression of the transformed log-periodogram, \(\log I_k\), on the regressors \(-2 \log(\lambda_k)\) for \(k = 1, 2, \dots, m\), and then adding one to the estimate: \\  
	\begin{equation*}
		\hat{d}_{WBLP}  = \frac{\sum_{k=1}^{m}\left(X_{k}-\overline{X}_{m}\right)	\log I_k }{\sum_{k=1}^{m}\left(X_{k}-\overline{X}_{m}\right)^2} + 1,
	\end{equation*}
	where $\overline{X}_{m} = \frac{1}{m}\sum_{k=1}^{m}X_{k}.$\\
	
	\noindent In this regression, the term \( \log\left(\frac{s(\lambda_k)}{s(0)}\right) = \log\left(\frac{f_u(\lambda_k)}{f_u(0)}\right) + \log\left(\frac{h(\lambda_k)}{h(0)}\right) \) plays a critical role in the asymptotic bias. Through a Taylor expansion around \( \lambda = 0 \), this term is approximately expressed as:  
	\[
	\log\left(\frac{s(\lambda_k)}{s(0)}\right) = \frac{1}{2} \frac{s''(0)}{s(0)} \lambda_k^2 + O(\lambda_k^4).
	\]
	Under Assumptions 1, 2, and 4, \cite{lee1} demonstrates that:  
	\[
	\begin{aligned}
		(a) \quad & \mathbb{E}[\hat{d}_{WBLP} - d] = -\frac{2\pi^2}{9} \frac{s''(0)}{s(0)} \frac{m^2}{n^2}(1 + o(1)) + O\left(\frac{m^4}{n^4}\right) + O\left(\frac{\log^3 m}{m}\right), \\ 
		(b) \quad & \text{Var}(\hat{d}_{WBLP}) = \frac{\pi^2}{24m} + o\left(\frac{1}{m}\right).
	\end{aligned}
	\]
	
	\noindent These results show that the wavelet-based log-periodogram estimator $\hat{d}_{WBLP}$ is consistent for \( d \in (0, 0.5) \), and that its variance retains a form similar to that observed in the stationary Gaussian case (for more details, see \cite{hdb}). \\
	
	\noindent The mean squared error (\(\text{MSE}\)) for the estimation of the parameter $\hat{d}_{WBLP}$ is given by:  
	\[
	\text{MSE}(\hat{d}_{WBLP}) = \left(\frac{2\pi^2}{9} \frac{s''(0)}{s(0)}\right)^2 \frac{m^4}{n^4}(1 + o(1)) + O\left(\frac{m^3 \log^3 m}{n^4}\right) + \frac{\pi^2}{24m}(1 + o(1)).
	\]  
	
	\noindent By exploiting this expression, the optimal rate for \(m\) is determined as:  
	\[
	m^* = \left[0.4634 \cdot \left(\frac{s(0)}{s''(0)}\right)^{2/5} n^{4/5} \right],
	\]
	where \( [x] \) denotes the integer part of \( x \). \\
	
	\noindent Although the terms \( s(0) \) and \( s''(0) \) are unknown, this relation implies that:  
	\[
	\text{MSE}(\hat{d}_{WBLP}) = O(n^{-4/5}),
	\]
	which corresponds to the same convergence rate as that of the $\hat{d}_{GPH}$ estimator in the stationary case. \\ 
	
	\noindent \cite{lee1} also establishes that, under the assumption \( m = O(n^{4/5}) \), the estimator \(\hat{d}_{WBLP}\) follows an asymptotic normal distribution:  
	\[
	\sqrt{m}(\hat{d}_{WBLP} - d) \xrightarrow{D} \mathcal{N}(0, \frac{\pi^2}{24}) \quad \text{as } n \to \infty.
	\]  
	
	\noindent These results, based on the works of \cite{hdb} and \cite{andrew}, demonstrate the relevance of this approach in the context of semi-parametric estimation for long-memory time series. \\
	
	\noindent After examining the method proposed by \cite{lee1}, which uses the wavelet transform to refine the estimation of long-memory parameters, we now turn to an alternative approach aimed at improving this estimation.
	\section{Log Regression via A New Periodogram} \label{sec_04}
	In this section, we propose a new periodogram along with an innovative method for estimating the long-memory parameter in the wavelet domain. This method aims to improve the estimation of this parameter in long-memory stochastic volatility (LMSV) models. We will denote our periodogram by NKK periodogram and our estimator by NKK estimator.
	
	\subsection{NKK Periodogram}
	
	\noindent Consider a time series \(\{Y_{1},...,Y_{n}\}\), where $n$ denotes the sample size. By applying the discrete wavelet transform to $\left\{Y_t\right\}$, we obtain the wavelet coefficients given by formula \eqref{coeff_ondelette}. The periodogram for the wavelet transform at scale $j$, introduced by \cite{lee1}, is defined as:  
	\begin{equation*}
		I_k^{(j)} = \frac{1}{2\pi n} \sum_{q=0}^{2^J-1} \left| w_{jq} \exp(i\lambda_k q) \right|^2, \quad k = 1, 2, \dots, m,
	\end{equation*}
	where \( \lambda_k = \frac{2\pi k}{n} \). The integers \( j \) and \( q \) represent the scale and translation parameters, respectively, where \( j = 0, 1, \dots, J \) and \( q = 0, 1, \dots, 2^j - 1 \). The finest (maximum) scale is fixed at \( J \), where \( n = 2^J \). \\
	
	\noindent This periodogram \( I_k^{(j)} \) at scale \( j \) can also be expressed in a linear regression framework as follows:  
	\begin{equation*}
		I_k^{(j)} = \frac{n}{8\pi} \|\tilde{\beta}^{(j)}(\lambda_{k})\|^2,
	\end{equation*}
	where \(\|\cdot\|\) denotes the Euclidean (\(\ell_2\)) norm of vectors, and \(\tilde{\beta}^{(j)}(\lambda_{k})\) represents the solution to the classical least squares problem:  
	\begin{equation*}
		\tilde{\beta}^{(j)}(\lambda_k) = \arg \min_{\beta \in \mathbb{R}^2} \sum_{q=0}^{2^{J}-1} \left| w_{jq} - h_q^\top(\lambda_k)\beta \right|^2,
	\end{equation*}
	with \(h_q(\lambda_k) = [\cos(\lambda_k q), \sin(\lambda_k q)]^\top\) as the harmonic regressor vector. \\
	
	\noindent To improve the robustness compared to the periodogram proposed by \cite{lee1}, we suggest replacing the least squares criterion with the least absolute deviations (LAD) criterion, as done by \cite{li} with the ordinary periodogram.  
	Thus, the new regression coefficient is expressed as:  
	\begin{equation*}
		\hat{\beta}_{LAD}^{(j)}(\lambda_{k}) = \arg \min_{\beta \in \mathbb{R}^2} \sum_{q=0}^{2^{J}-1} \left| w_{jq} - h_q^\top(\lambda_k)\beta \right|.
	\end{equation*}
	\noindent \begin{definition}
		Let $\left\{Y_{1},...,Y_{n}\right\}$ be a time series of size \( n \), and let \(\{w_{jq}\}\) be the set of coefficients obtained after applying the discrete wavelet transform to this time series, where \( j = 0, 1, \dots, J \) denotes the scales and \( q = 0, 1, \dots, 2^j - 1 \) the translation parameters. The NKK periodogram at scale \( j \) is defined as:
		\begin{equation*}
			N_k = N_k^{(j)} = \frac{n}{8\pi} \|\hat{\beta}_{LAD}^{(j)}(\lambda_{k})\|^2,
		\end{equation*}
		where \(\hat{\beta}_{LAD}^{(j)}(\lambda_{k})\) represents the solution to the  LAD minimization problem applied to the wavelet coefficients \(\{w_{jq}\}\) at scale \( j \).
	\end{definition}
	
	\noindent This periodogram is expected to inherit the robustness benefits associated with the LAD regression, as discussed by \cite{bloom}. \\
	
	\noindent With the NKK periodogram defined, we now proceed to introduce the NKK estimator derived from it.
	
	\subsection{NKK Estimator}
	In this subsection, we propose a method for estimating the long-memory parameter by combining the log periodogram regression and the NKK periodogram. This approach aims to leverage the robustness of the LAD criterion while exploiting the multi-resolution properties of wavelets. \\
	
	\noindent We adopt the same procedure described in \cite{lee1}, but replace the wavelet-based periodogram \(I_k^{(j)}\) with the NKK periodogram \(N_k^{(j)}\) at scale \(j\). By substituting this modification into Equation \eqref{wavelet_LP}, we obtain:
	\begin{equation*}
		\log N_k = \alpha + \gamma X_k + \log\left[\frac{s\left(\lambda_k\right)}{s\left(0\right)}\right]  + \varepsilon^{N}_k, \quad k = 1, 2, \dots, m, 
	\end{equation*}
	where:
	\[
	\alpha = \log C_J + \log s(0), \quad \gamma = d - 1, \quad X_k = -2\log(\lambda_k), \quad \varepsilon^{N}_k = \log\left[\frac{N_{k}}{f_{w}(\lambda_{k})}\right].
	\] 
	
	\noindent We establish the necessary assumption for defining our estimator.\\
	
	\begin{assump}\label{asssump4}
		\begin{itemize}
			\item [(a)] For low frequencies, the term \(\log\left[\frac{s\left(\lambda_k\right)}{s\left(0\right)}\right]\) is negligible.
			\item [(b)] The random variables \(\log \left[\frac{N_{k}}{f_{w}(\lambda_{k})}\right]\), \(k = 1, 2, \dots, m\), are asymptotically \(i.i.d.\)
		\end{itemize}
	\end{assump}
	\noindent We now give on estimator in the following theorem.\\
	
	\noindent \begin{theorem}
		Assuming that Assumptions \ref{assump2.1}, \ref{assump3.1}, and \ref{asssump4} hold, the NKK estimator \(\hat{d}_{NKK}\) is given by:
		\begin{equation*}
			\hat{d}_{NKK} = \frac{\sum_{k=1}^{m} \left(X_k - \overline{X}_m\right) \log N_k}{\sum_{k=1}^{m} \left(X_k - \overline{X}_m\right)^2} + 1,
		\end{equation*}
		where \(\overline{X}_m = \frac{1}{m}\sum_{k=1}^{m} X_k\).
	\end{theorem}
	
	\begin{proof}
		Under Assumptions \ref{assump2.1} and \ref{assump3.1}, the approximate Gaussian spectral representation of the discrete wavelet transform around the zero frequency is given by Equation (\ref{densite_avec_ond}):  
		\begin{equation*}
			f_w(\lambda) = C_J \vert\lambda\vert^{-2(d-\nu)} f_{u}(\lambda) h(\lambda) \quad \text{as} \quad \lambda \to 0, \quad \text{for} \quad d \in (0, 0.5), \label{densite_avec_ond new}
		\end{equation*}  
		where $C_J = 2^{-J(1+2\nu)}$ and \( h(\lambda) = b^2(\lambda) \). \\
		
		\noindent Using the Haar wavelet (\(\nu = 1\)) in the wavelet transform, and letting \(s(\lambda) = f_u(\lambda)h(\lambda)\), we apply the logarithm to Equation (\ref{densite_avec_ond}) to obtain:  
		\begin{equation}
			\log f_{w}\left(\lambda\right) = \log C_J + \log s\left(0\right) -2(d-1)\log  \left(\lambda\right) + \log\left[\frac{s\left(\lambda\right)}{s\left(0\right)}\right]. \label{33}
		\end{equation}  
		
		\noindent Adding \(\log N_{k}\), evaluated at the Fourier frequencies \(\lambda_{k}\) for \(k = 1, 2, \dots, m\), to both sides of Equation (\ref{33}), we obtain:  
		\begin{equation}
			\log  N_{k} = \log C_J + \log s\left(0\right) -2(d-1)\log  \left(\lambda_k\right)  + \log\left[\frac{s\left(\lambda_k\right)}{s\left(0\right)}\right] + \log\left[\frac{N_{k}}{f_{w}(\lambda_{k})}\right]. \label{44}
		\end{equation}  
		
		\noindent Under Assumption \ref{asssump4}, we can rewrite (\ref{44}) as a linear regression equation:  
		\begin{equation*}
			\log N_k = \alpha + \gamma X_k  + \varepsilon^{N}_k, \quad k = 1, 2, \dots, m, 
		\end{equation*}  
		where:  
		\[ 
		\alpha = \log C_J + \log s(0), \quad \gamma = d - 1, \quad Y_k = \log N_k, \quad X_k = -2\log(\lambda_k), \quad  \varepsilon^{N}_k = \log\left[\frac{N_{k}}{f_{w}(\lambda_{k})}\right].
		\]  
		
		\noindent Solving this regression using least squares yields the following estimator for \(\gamma\):  
		\begin{equation*}
			\hat{\gamma} = \frac{\sum_{k=1}^{m} \left(X_k - \overline{X}_m\right)(Y_k - \overline{Y}_m)}{\sum_{k=1}^{m} \left(X_k - \overline{X}_m\right)^2}, 
		\end{equation*}
		where \(\overline{X}_m = \frac{1}{m}\sum_{k=1}^{m} X_k\) and \(\overline{Y}_m = \frac{1}{m}\sum_{k=1}^{m} Y_k\).\\  
		
		\noindent Notice that in this expression, $\overline{Y}_m$ represents the mean of $\log N_k$. A key property of centered deviations is that the sum of deviations from the mean is always zero:  
		\begin{equation*}
			\sum_{k=1}^{m} \left(X_k - \overline{X}_m\right) = 0.
		\end{equation*}
		
		\noindent Using this property, we can expand the numerator of \(\hat{\gamma}\) as follows:  
		\begin{equation*}
			\sum_{k=1}^{m} \left(X_k - \overline{X}_m\right)(Y_k - \overline{Y}_m) = \sum_{k=1}^{m} Y_k (X_k - \overline{X}_m) - \overline{Y}_m \sum_{k=1}^{m} \left(X_k - \overline{X}_m\right).
		\end{equation*}
		
		\noindent Since \(\sum_{k=1}^{m} \left(X_k - \overline{X}_m\right) = 0\), the second term in the expansion becomes zero. Thus, the numerator simplifies to:  
		\begin{equation*}
			\sum_{k=1}^{m} \left(X_k - \overline{X}_m\right)(Y_k - \overline{Y}_m) = \sum_{k=1}^{m} Y_k (X_k - \overline{X}_m).
		\end{equation*}
		
		\noindent Substituting this result back into the expression for \(\hat{\gamma}\), we obtain:  
		\begin{equation*}
			\hat{\gamma} = \frac{\sum_{k=1}^{m} (X_k - \overline{X}_m) \log N_k}{\sum_{k=1}^{m} (X_k - \overline{X}_m)^2}.
		\end{equation*}
		
		\noindent Therefore, the final estimator for the memory parameter is given by:  
		\begin{equation*}
			\hat{d}_{NKK} = \hat{\gamma} + 1.
		\end{equation*}
	\end{proof}
	
	\noindent Having defined the NKK periodogram and its corresponding estimator, we now proceed to an empirical evaluation of its performance through numerical simulations, comparing it to existing methods.
	
	\section{Simulation Study}\label{sec_05}
	In this section, we compare the finite-sample performance of the GPH estimator, the WBLP estimator, and the NKK estimator. \\
	
	\noindent To facilitate the comparison, we simulate the process from model (\ref{LMSV}) using, as in \cite{lee1}, the same data generator: an ARFIMA(1, \(d\), 0) process defined as:
	\begin{equation*}
		(1 - \phi)(1 - L)^d Z_t = \epsilon_t, \label{eq:ARFIMA}
	\end{equation*}
	where \(\phi\) is the autoregressive parameter and \(\epsilon_t\) is a sequence of independent and identically distributed (\textit{i.i.d.}) random variables with variance \(\sigma_\epsilon^2\). The fractional process \(I(d)\), denoted as \(\{Z_t\}_{t=1}^n\), is generated using the series:
	\begin{equation*}
		Z_t = \sum_{k=0}^{t-1} \frac{(d)_k}{k!} u_{t-k}, \quad \text{with } (d)_k = d(d+1)\cdots(d+k-1),
	\end{equation*}
	where \(u_t \sim \mathcal{N}(0, 1)\) is Gaussian white noise. \\
	
	\noindent For Figures \ref{fig1} and \ref{fig2}, the sample size is fixed at \(n = 1024\). The value of \(\sigma_\epsilon^2\) is set to 0.37, following the work of \cite{DeoHurvich2001} and \cite{lee1}. We consider the combinations \((d, \phi) \in \{(0.2, 0.4), (0.3, 0.5)\}\).Other combinations exhibit similar qualitative results. Regarding the number of frequencies \(m\) used in the regression, we include values of \(m\) ranging from \([n^{0.3}]\) to \([n^{0.8}]\), where \([x]\) denotes the integer part of \(x\). We focus on the variation of the Mean Squared Error (MSE), which is one of the most used metrics for assessing estimator performance. For each value of \(m\), we perform 1000 iterations to obtain a robust estimate of the MSE. \\
	
	\noindent For both the WBLP and NKK estimators, we use the Haar wavelet for the discrete wavelet transform. The entire scale \(j\) is fixed at the maximum scale \(J\). With \(n = 1024 = 2^{10}\), we set \(J = 10\), which generates the series of wavelet coefficients \(\{w_{Jq}, q = 0, 1, \dots, 2^{10} - 1\}\). \\
	
	\noindent The results presented in Figures \ref{fig1} and \ref{fig2} clearly show a similar trend in the variation of the MSE with the number of frequencies \(m\). When \(m\) is small, all three estimators exhibit high mean squared errors, reflecting instability due to an insufficient number of frequencies used in the regression. However, as \(m\) increases, the MSE decreases rapidly, reaching a relatively stable value. This transition occurs around \(m \approx 40\) for all three methods. \\
	
	\noindent In particular, the results highlight performance differences among the three methods studied:
	\begin{itemize}
		\item The NKK estimator achieves the lowest MSE over most of the range of \(m\), indicating superior robustness and accuracy, except for very small values of \(m\), where the GPH estimator exhibits a lower MSE.
		\item The WBLP estimator closely approaches the performance of the NKK estimator and provides comparable results for large values of \(m\). This method serves as an interesting alternative, particularly in contexts where increased robustness to non-Gaussian noise is critical.
		\item The GPH estimator, while effective for small values of \(m\), exhibits a slightly higher MSE for medium and large values of \(m\), reflecting its increased sensitivity to low frequencies.
	\end{itemize}
	\begin{figure}[h!]
		\centering
		\includegraphics[width=0.9\textwidth]{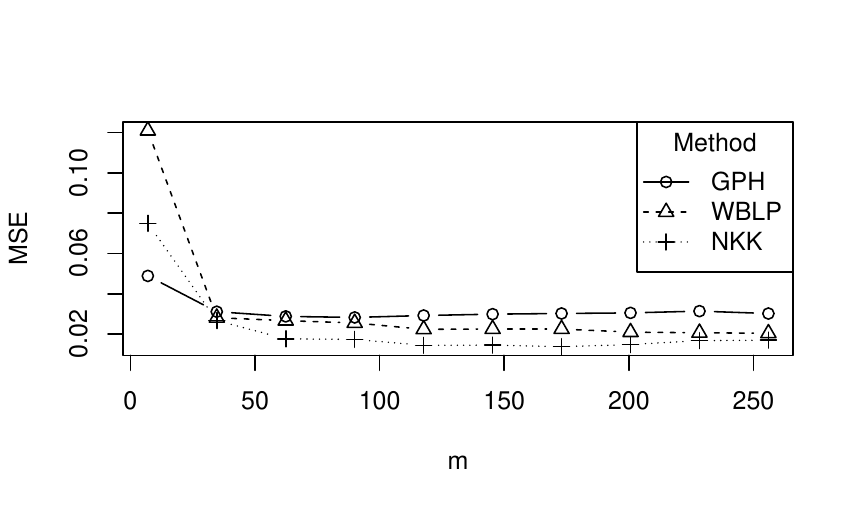}
		\caption{Mean squared error (MSE) as a function of the number of frequencies $m$ for $n = 1024$, $(d, \phi) = (0.2, 0.4)$.}
		\label{fig1}
	\end{figure}
	\begin{figure}[h!]
		\centering
		\includegraphics[width=0.8\textwidth]{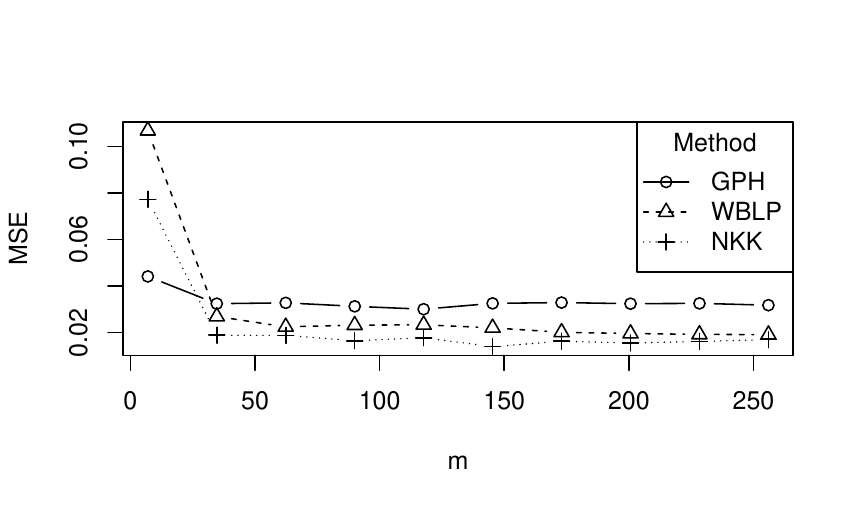}
		\caption{Mean squared error (MSE) as a function of the number of frequencies $m$ for $n = 1024$, $(d, \phi) = (0.3, 0.5)$.}
		\label{fig2}
	\end{figure}
	\begin{figure}[h!]
		\centering
		\includegraphics[width=0.8\textwidth]{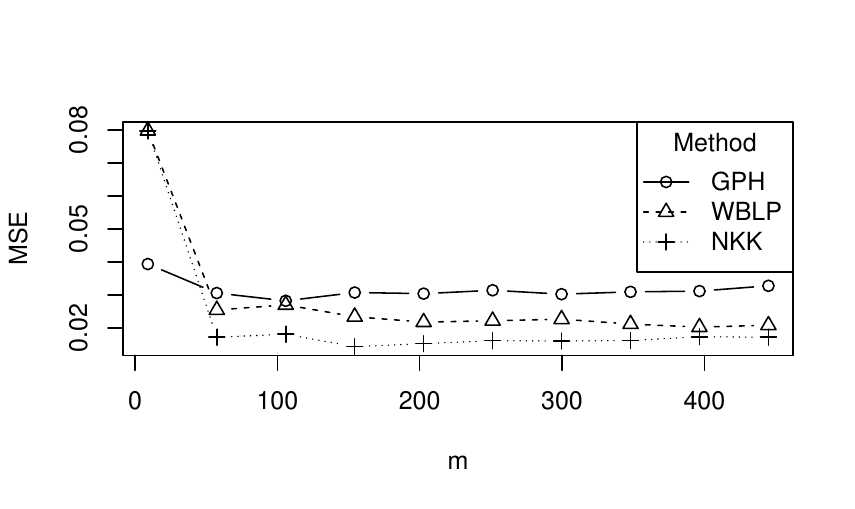}
		\caption{Mean squared error (MSE) as a function of the number of frequencies $m$ for $n = 2048$, $(d, \phi) = (0.3, 0.4)$.}
		\label{fig3}
	\end{figure}
	
	\noindent These results confirm the importance of appropriately choosing $m$ in the estimation and show that the NKK and WBLP estimators outperform the classical GPH estimator, particularly in the simulated scenario studied. The rapid convergence to a stable MSE for these two methods also makes them attractive for applications with moderate sample sizes.\\
	
	\noindent Next, in Figure \ref{fig3}, we increase the sample size to $n = 2048$, fix $(d, \phi) = (0.3, 0.4)$, and $\sigma_\epsilon^2 = 0.37$. With $n = 2048 = 2^{11}$, we set $J = 11$, generating the series of wavelet coefficients $\{w_{Jq}, q = 0, 1, \dots, 2^{11} - 1\}$. The obtained results exhibit similar trends to those observed previously. However, the decrease in the MSE as $m$ increases is slightly slower due to the larger sample size. Once again, the NKK estimator stands out for its superior performance, followed by the WBLP estimator, while the GPH estimator lags slightly behind.\\
	
	\noindent The results of the numerical simulations highlight the superiority of the NKK estimator across various scenarios, particularly for moderate sample sizes and in the presence of non-Gaussian noise.
	
	\section{Concluding remarks}\label{sec_06}
	In this paper, we explored new perspectives for estimating the long-memory parameter \(d\) by introducing the NKK periodogram, a robust extension of the wavelet-based log-periodogram. This method is based on least absolute deviations applied to wavelet coefficients, offering improved resistance to non-Gaussian noise spectra and enhancing estimation accuracy. \\
	
	\noindent The simulation results highlighted the superiority of the NKK estimator over existing approaches, particularly the GPH estimator of \cite{gph} and the WBLP estimator proposed by \cite{lee1}. The NKK estimator stands out for its robustness and precision over a wide range of frequencies, especially in contexts where noise presence can undermine the performance of classical methods. \\
	
	\noindent The integration of wavelets, thanks to their multi-resolution properties and their frequency and time localization, opens new perspectives for the analysis of complex time series, allowing for better modeling of long memory while mitigating the impact of perturbations and complex dynamics.\\

	\bibliographystyle{amsplain}

\end{document}